\documentclass[onecolumn]{IEEEtran}
\usepackage{mathpazo}
\usepackage{times}

\usepackage{amsmath}
\usepackage{amsfonts}
\usepackage{latexsym}
\usepackage{amssymb}
\usepackage{upref}
\usepackage{theorem}
\usepackage{graphicx}
\usepackage{psfrag}
\usepackage{cite}
\usepackage{epsfig}
\usepackage{color}
\usepackage{graphicx}
\usepackage[left=1in, right=1in, top=1.5in, bottom=1.5in]{geometry}

\theoremstyle{plain}
\theorembodyfont{\normalfont\slshape}

\newtheorem{thm}{Theorem$\!$}

\newtheorem{clm}[thm]{Claim$\!$}

\newtheorem{lem}[thm]{Lemma$\!$}

\newtheorem{prop}[thm]{Proposition$\!$}

\newtheorem{cor}[thm]{Corollary$\!$}

\newtheorem{defn}[thm]{Definition$\!$}

\newtheorem{xmpl}{Example$\!$}

\newtheorem{cnstr}{Construction$\!$}

\newcounter{enumrom}
\renewcommand{\theenumrom}{(\roman{enumrom})}

\makeatletter
\renewcommand{\@endtheorem}{\endtrivlist}
\makeatother

\makeatletter
\renewcommand{\thefigure}{{\@arabic\c@figure}}
\renewcommand{\fnum@figure}{{\bf Figure\,\thefigure}}
\makeatother

\newcommand{\cC}{{\cal C}}
\newcommand{\cD}{{\cal D}}

\DeclareMathOperator{\spun}{span}

\begin{document}


\title{\textbf{Explicit MDS Codes for Optimal Repair Bandwidth}
}

\author{\IEEEauthorblockN{Zhiying Wang, Itzhak Tamo, and Jehoshua Bruck\\}
}


\maketitle

\begin{abstract}
MDS codes are erasure-correcting codes that can correct the maximum number of erasures  for a given number of redundancy or parity symbols. If an MDS code has $r$ parities and no more than $r$ erasures occur, then by transmitting all the remaining data in the code, the original information can be recovered. However, it was shown that in order to recover a single symbol erasure, only a fraction of $1/r$ of the information needs to be transmitted. This fraction is called the \emph{repair bandwidth (fraction)}. Explicit code constructions were given in previous works. If we view each symbol in the code as a vector or a column over some field, then the code forms a 2D array and such codes are especially widely used in storage systems. In this paper, we address the following question: given the length of the column $l$, number of parities $r$, can we construct high-rate MDS array codes with optimal repair bandwidth of $1/r$, whose code length is as long as possible?
In this paper, we give code constructions such that the code length is $(r+1)\log_r l$.
\end{abstract}



\let\thefootnote\relax\footnotetext{Z. Wang is with Department of Electrical Engineering, Stanford University, Stanford, CA 94305, USA (email: zhiyingw@stanford.edu).

I. Tamo is with Department of Electrical and Computer Engineering and Institute for Systems Research, University of Maryland, College Park, MD 20742, USA (email: tamo@umd.edu).

J. Bruck is with Department of Electrical Engineering, California Institute of Technology,
Pasadena, CA 91125, USA (email: bruck@caltech.edu).
}

%
%
\section{Introduction}
%
%
MDS (maximum distance separable) codes are optimal error-correcting codes in the sense that they have the largest minimum distance for a given number of parity symbols. If each symbol is a vector or a column, we call such a code an MDS array code (e.g. \cite{Blaum, B-code, x-code, RDP-code,Plank}). In (distributed) storage systems, each column is usually stored in a different disk, and MDS array codes are widely used to protect data against erasures due to their error correction ability and low computational complexity. In this paper, we call each symbol a column or a node, and the column length, or the vector size of a symbol, is denoted by $l$.

If an MDS code has $r$ parities, then it can correct up to $r$ erasures of entire columns. In this paper, we not only would like to recover the erasures, but also care about the efficiency in recovery: what is the fraction of the remaining data transmitted in order to correct the erasures?  We call this fraction the \emph{repair bandwidth (fraction)}. For example, if $r$ erasures happen, it is obvious that we have to transmit all of the remaining information, therefore, the fraction is $1$. For a single erasure it was shown in \cite{Dimakis} (which also formulated the repair problem) that this fraction is actually lower bounded by $1/r$. In the general case, it was shown in \cite{Tamo2} that
when $e \le r$ nodes are erased, then the repair bandwidth is lower bounded by $e/r$.
Since the repair of information is much more crucial than redundancy, and we study mainly high-rate codes, we will focus on the optimal repair of information or systematic nodes. Moreover, since single erasure is the most common scenario in practice, we assume $e=1$. Thus, in this paper a code is said to have an optimal repair if this bound of $1/r$ is achieved for the repair of \emph{any} of its systematic nodes.
For example, in Figure \ref{fig1}, we show an MDS code with $4$ systematic nodes, $r=2$ parity nodes, and column length $l=2$. One can check that this code can correct any two erasures, therefore it is an MDS code. In order to repair any systematic node, only $1/r=1/2$ fraction of the remaining information is transmitted. Thus this code has optimal repair.

\begin{figure*}
 \centering
 \begin{tabular}{|c|c|c|c|c|c|}
   \hline
   N1 & N2 & N3 & N4 & N5 & N6 \\
   \hline
   $a$ & $b$ & $c$ & $d$ & $a+b+c+d$ & $2a+w+2b+3c+d$ \\
   \hline
   $w$ & $x$ & $y$ & $z$ & $w+x+y+z$ & $3w+b+3x+2y+z$ \\
   \hline
 \end{tabular}
    \caption{(n=6,k=4,l=2) MDS code over finite field $\mathbb{F}_4$, and we use $\{0,1,2,3\}$ to represent its elements. The first $4$ nodes are systematic and the last $2$ are parities. To repair $N1$ transmit the first row from every remaining node. To repair $N2$ transmit the second row. To repair $N3$ transmit the sum of both rows. And to repair $N4$ transmit the sum of the first row and $2$ times the second row from nodes $N1,N2,N3,N5$, and the sum of the first row and $3$ times the second row from node $N6$.}\label{fig1}
\end{figure*}

In \cite{Dimakis-interference-alignment, Rashmi, Suh-alignment, Kumar2009, Wu07} codes achieving the repair bandwidth lower bound were studied where the number of systematic nodes is less than the number of parity nodes (low code rate). For arbitrary code rate, \cite{Cadambe2013} proved that the lower bound is asymptotically achievable when the column length $l$ goes to infinity. And \cite{viveck, viveck3, Papailiopoulos, Wang, Tamo2} studied codes with more systematic nodes than parity nodes (high code rate) and finite $l$, and achieved the lower bound of the repair bandwidth. If we are interested in the \emph{code length} $k$, i.e., the number of systematic nodes given $l$, low-rate codes have a linear code length $l+1$ \cite{Kumar2009,Suh-alignment}; on the other hand, high-rate constructions are relatively short. For example, suppose that we have 2 parity nodes, then the number of systematic nodes is only $\log_2 l$ in all of the constructions, except for \cite{viveck3} it is $2\log_2 l$. In \cite{Tamo3} it is shown that an upper bound for the code length is $k \le 1+l \binom{l}{l/2}$, and the bound is further tightened to $k \le 2 (\log_2 l) (\log_2 l + 1)+1$ in \cite{Goparaju}. But the tightness of the above bounds is not known. It is obvious that there is a gap between this upper bound and the constructed codes.

Besides bandwidth which corresponds to transmission incurred during repair, we are also interested in \emph{access}. It is defined as the fraction of data read in the surviving nodes in order to repair an erasure. Access is an important metric because it affects the disk I/O operations and hence the speed and complexity in repair. Since a transmitted symbol can be functions of many read symbols, we know that access is no less than $1/r$. For example, in Figure \ref{fig1} the repair of node $N1$ reads and transmits only the first row, so the repair bandwidth and access are both $1/2$. However, the repair of node $N3$ requires reading both rows, so the access is $1$. Moreover, we define \emph{update} as the number of necessary writes if a symbol is rewritten in the code. This metric is important when blocks of the stored data is frequently updated. In Figure \ref{fig1} symbol $a$ appears 3 times in the code and therefore its update is 3, while symbol $w$ has update 4. For an MDS code with $r$ parities, it is not difficult to see that the update should be no less than $r+1$ for each symbol. And we say that a code achieving this bound is optimal update.

The main contribution of this paper is as follows: 
\begin{enumerate}
\item
We construct high-rate codes with $r$ parity nodes and $(r+1) \log_r l$ systematic nodes. In particular, with $2$ parity nodes we get a code length of $3 \log_2 l$, moreover, this code uses a finite field of size $1+2\log_2 l$. 
\item
We rigorously state some sufficient properties of linear optimal repair codes (similar results also seen in \cite{Suh-alignment, Cadambe2013,Kumar2009}), and thus enable explicit code construction and simplify proofs of optimality. 
\item
We design optimal-update codes with $2$ parities and $2 \log_2 l$ systematic nodes. This construction exceeds the upper bound of $k \le \log_2 l$ given by \cite{Tamo3} for optimal-update and diagonal encoding matrices. Diagonal encoding matrices means that the encoding are done only within each row in the array code. However our construction allows mixing of different rows in encoding. As a result, we can see a fundamental difference between these two types of codes.
\item
We construct a family of codes that further reduces the access compared to the proposed optimal-bandwidth code. We use a technique that transforms a linear code to an equivalent one through block-diagonal matrix. This technique can be applied to an arbitrary optimal-bandwidth code and therefore can be a useful tool for future codes as well.
\end{enumerate}

Even though our construction with  $(r+1) \log_r l$ systematic nodes is additive improvement for code length compared to \cite{viveck3}, where the code length is $r \log_r l$, we point out here a few advantages of our work. Through the sufficient properties of optimal repair codes, we are then able to explicitly write the code generating matrix in terms of eigenspaces and eigenvalues, whereas \cite{viveck3} constructed codes recursively by Kronecker product of matrices and multiplication of permutation matrices. Moreover, our technique eigenspaces inspired new code constructions in recent work \cite{Li}. Also in \cite{viveck3} the code requires a large enough finite field. But in our construction the finite field size is specified for the $2$ parity case, and therefore can be practical for distributed storage applications.

The rest of the paper is organized as follows: in Section \ref{sec2} we will formally introduce the repair bandwidth and the code length problem. In Section \ref{sec3} codes with $r$ parity nodes are constructed, and we show that the code length is $(r+1) \log_r l$. We will show an optimal-update code with $2\log_2 l$ systematic nodes and 2 parity nodes in Section \ref{sec7}, and discuss about reducing the access ratio in Section \ref{sec6}. Finally we conclude in Section \ref{sec5}.

%
%
\section{Problem Settings} \label{sec2}
%
%

We define in this section the array code by specifying the encoding, repair, and reconstruction processes.
 
\subsection{Encoding}
An $(n,k,l)$ MDS array code is an $(n-k)$-erasure-correcting code such that each symbol is a column of length $l$. The number of systematic symbols is $k$ and the number of parity symbols is $r=n-k$. We call each symbol a column or a node, and $k$ the \emph{code length}. We assume that the code is systematic, hence the first $k$ nodes of the code are information or systematic nodes,  and the last $r$ nodes are parity or redundancy nodes.

Suppose the columns of the code are $C_1,C_2,\dots,C_{n}$, each being a column vector in $\mathbb{F}^l$, for some finite field $\mathbb{F}$. We assume that the parity nodes are a linear function of the information nodes. Namely, for $i=1,...,r$, parity node $k+i$ is defined by the invertible \emph{encoding matrices} of size $l$ $A_{i,j}$, $j=1,...,k$ as follows
$$C_{k+i}=\sum_{j=1}^k A_{i,j}C_{j}.$$
For example, in Figure \ref{fig1}, the encoding matrices are $A_{1,j}=I$ for all $j=1,..,4$, and
$$A_{2,1}=\left(
    \begin{array}{cc}
      t & 1 \\
      0 & t+1 \\
    \end{array}
  \right),A_{2,2}=
  \left(
    \begin{array}{cc}
      t & 0 \\
      1 & t+1 \\
    \end{array}
  \right),
  A_{2,3}=\left(
    \begin{array}{cc}
      t+1 & 0 \\
      0 & t \\
    \end{array}
  \right),
  A_{2,4}=\left(
    \begin{array}{cc}
      1 & 0 \\
      0 & 1 \\
    \end{array}
  \right).
$$
Here the finite field is $\mathbb{F}_4$ generated by the irreducible polynomial $t^2+t+1$, and in the table $t, t+1$ are written as $2,3$, respectively.
In our constructions, we require that $A_{1,j}=I$ for all $j \in [k]$. Hence the first parity is the row sum of the information array. Even though this assumption is not necessarily true for an arbitrary linear MDS array code, it can be shown that any linear code can be equivalently transformed into one with such encoding matrices \cite{Tamo3}.

\subsection{Repair}
Suppose a code has optimal repair for any systematic node $i$, $i \in [k]$, meaning only a fraction of $1/r$ data is transmitted in order to repair a node erasure. When a systematic node $i$ is erased, we are going to use size $l/r \times l$ matrices $S_{i,j}$, $j \neq i, j \in [n]$, to repair the node: From a surviving node $j$, we are going to compute and transmit $S_{i,j} C_j$, which is only $1/r$ of the information in this node. 

It was shown in \cite{Tamo3} that we can further simplify our repair strategy of node $i$ and assume by equivalent transformation of the encoding matrices that
\begin{equation}\label{eq21}
  S_{i,j}=S_i, \textrm{ for all } j \neq i, j \in [n].
\end{equation}
\textbf{Notation:} By abuse of notations, we write $S_i, S_i A_{t,j}$ both to denote both the matrices of size $l/r \times l$ and the subspaces spanned by their rows.

In the following we show necessary and sufficient conditions for optimal repair.
\begin{clm} \cite{Tamo3}
Optimal repair of a systematic node $i$ is equivalent to the following \textbf{subspace property}: 
There exist a matrix $S_{i}$ of size $l/r \times l$, such that for all $j \neq i, j \in [k], t \in [r]$,
\begin{align} 
S_i &= S_i A_{t,j}, \label{eq8}\\
\sum_{t=1}^{r} S_i A_{t,i} &= \mathbb{F}^l  \label{eq9}
\end{align}
\end{clm}
Here the equalities are defined on the row spans instead of the matrices, and the sum of two subspaces $A,B$ is defines as $A+B=\{a+b: a \in A, b \in B\}$. Obviously, in \eqref{eq9} the dimension of each subspace $S_{i}A_{t,i}$ is no more than $l/r$, and the sum of $r$ such subspaces has dimension no more than $l$. This means that these subspaces intersect only on the zero vector. Therefore, the sum is actually a direct sum of the subspaces, and matrix $S_{i}$ has full rank $l/r$.

\begin{IEEEproof}[Sketch of proof:]
Suppose the code has optimal repair bandwidth, then we need to transmit $l/r$ elements from each surviving column. Suppose we transmit $S_{i}C_j$ from a systematic node $j \neq i, j \in [k]$, and $S_{i}C_{k+t} = \sum_{z=1}^{k} S_{i}A_{t,z}C_z$ from a parity node $k+t \in [k+1,k+r]$. Our goal is to recover $C_i$ and cancel out all $C_j$, $j \neq i, j\in[k]$.
In order to cancel out $C_j$, \eqref{eq8} must be satisfied. In order to solve $C_i$, all equations related to $C_i$ must have full rank $l$, so \eqref{eq9} is satisfied. One the other hand, if \eqref{eq8} \eqref{eq9} are satisfied, one can transmit $S_{i}C_j$ from each node $j$, $j \neq i, j \in [n]$ and optimally repair the node $i$. 
\end{IEEEproof}
Similar interference alignment technique was first introduced in \cite{Cadambe2013} for the repair problem. Also, \cite{Kumar2009} was the first to formally prove similar conditions. However, the reduction from distinct $S_{i,j}$ to identical $S_i$ for different values of $j$ was not known before.

Notice that if \eqref{eq8} is satisfied then $S_i$ is an invariant subspace of $A_{t,j}$ for any $t=1,...,r$ and $j\neq i$. If $A_{t,j}$ is diagonalizable then it is uniquely defined by its eigenspaces and eigenvalues. Moreover each of the invariant subspaces of $A_{t,j}$ has a basis composed of eigenvectors of $A_{t,j}$. Therefore, we will first focus on finding the proper encoding matrices, by defining their set of eigenspaces. These eigenspaces will uniquely define the set of invariant subspaces for each encoding matrix.  Then we will choose carefully the eigenvalue that corresponds to each eigenspace, in order to ensure the MDS property of the code.

For a general repair strategy, the subspaces $S_{i,j}, j \in [k]$ are not necessarily identical, and the general subspace property for optimal repair of a systematic node $i$ is: There exist matrices $S_{i,j}$, $j \neq i, j \in [n]$, all with size $l/r \times l$, such that for all $j \neq i, j \in [k], t \in [r]$,
\begin{align} 
S_{i,j} =& S_{i,k+t}A_{t,j}, \label{eq6}\\
\sum_{t=1}^{r}S_{i,k+t}A_{t,i}=&\mathbb{F}^l, \label{eq7}
\end{align}
where the equality is defined on the row spans instead of the matrices.

We mention here that if we use the simple repair strategy,\eqref{eq21} holds for all nodes $i$ with the possible exception of a single node. For instance see $N4$ in the following example.
However in the subsequent sections, we will shorten the code by one node if such exception exists and assume identical $S_{i,j}=S_i$ for all $i \in [k]$.

\begin{xmpl} \label{eg:fig1}
  In Figure \ref{fig1}, the matrices $S_i$ are
  $$S_1=(1, 0), S_2=(0, 1), S_3=(1,1).$$
  One can check that the subspace property \eqref{eq8}, \eqref{eq9} is satisfied for $i=1,2,3$. For instance, in order to repair systematic node $N3$, we need to transmit the sum of the elements from each node, which is equivalent to multiply each column by the matrix $S_3=(1,1)$. Note that $(1,1)$ is an eigenvector for $A_{t,j}$, $t=1,2, j=1,2,4$, hence we have $S_3 = S_3 A_{t,j}$, where the equality is between the subspaces. Furthermore, it is easy to check that $$S_3  \oplus S_3 A_{2,3}=\spun(1,1) \oplus \spun(t+1,t) = \mathbb{F}_4^2.$$  Node $N4$ is an exception, since the matrices $S_{4,j}$'s are not equal. In fact $S_{4,j}=(1,t)$ for $j=1,2,3,5$, and $S_{4,6}=(1,t+1)$.
\end{xmpl}

\subsection{Reconstruction}
If no more than $r$ of the nodes are erased, the MDS property requires that the entire information can be decoded from the remaining nodes. Usually this requirement can be satisfied by choosing proper coefficients in the encoding matrices over a large enough finite field. And in our constructions, it is satisfied by proper eigenvalues of the encoding matrices, as shown in the subsequent sections.

%
%
\section{Optiaml-Bandwidth Code Construction} \label{sec3}
%
%
In this section, we will construct a code with arbitrary number of parity nodes. Our code will have column length $l=r^m$, $k=(r+1)m$ systematic nodes, and $r$ parity nodes, for any positive integers $r, m$. We start with the construction description and proof for optimal repair, and then discuss the update and access complexity of the code, and at last argue that the entire information is reconstructible from any $r$ node erasures. 

\subsection{Construction}
We define the code, or equivalently the encoding matrices, in terms of their eigenspaces. We define $k$ diagonalizable matrices $A_1,...,A_k$ of order $l=r^m$, whose Jordan canonical form are diagonal matrices. Each matrix $A_i$ will have $r$ distinct non zero eigenvalues that correspond to $r$ eigenspaces, each of dimension $l/r=r^{m-1}$. The encoding matrix for parity node $k+s$, and systematic node $i$ is defined as
\begin{equation}
A_{s,i}=A_i^{s-1},s\in [r],i\in [k].
\label{tamtam}
\end{equation}
\textbf{Remark:}
\begin{enumerate}
	\item Each symbol in the first parity is simply a linear combination of the corresponding row, since $A_{1,i}=A_i^{1-1}=I$ for any $i$.
	\item Denote by $V_{i,0},V_{i,1},\dots,V_{i,r-1}$ the left eigenspaces of $A_i$ that correspond to eigenvalues $\lambda_{i,0},\lambda_{i,1},\dots,\lambda_{i,r-1}$, then $A_{s,i}$ has eigenvalues $\lambda_{i,0}^{s-1},\lambda_{i,1}^{s-1},\dots,\lambda_{i,r-1}^{s-1}$.
\end{enumerate}
By abuse of notations, $V_{i,u}$ represents both the eigenspace and the $l/r \times l$ matrix containing $l/r$ linearly independent eigenvectors.
Our construction will only focus on the matrix $A_{i}$.
Using the definition of the encoding matrices in \eqref{tamtam} the \textbf{subspace property} becomes
\begin{equation} \label{eq24}
S_i = S_i A_{j}, \forall j \neq i, j \in [k]
\end{equation}
\begin{equation}\label{eq25}
S_i + S_i A_{i} + S_i A_{i}^2 + \dots + S_i A_{i}^{r-1} = \mathbb{F}^l
\end{equation}
Hence, when a systematic node $i$ is erased, $i \in [k]$, we are going to use the subspace $S_i$ in order to optimally repair it. We term this subspace as the \emph{repairing subspace} of node $i$.

In the first step we will only define the eigenspaces of each matrix $A_i$ without specifying the eigenvalues. This will be enough to show the optimal repair property of the code. Then we will show that over a large finite field, there exist an assignment for the eigenvalues, that guarantees the MDS property as well.

Let $\{e_a: a=0,...,l-1\}$ be some basis of $\mathbb{F}^{l}$, for example, one can think of them as the standard basis vectors. The subscript $a$ is represented by its $r$-ary expansion, $a=(a_1,a_2,\dots,a_m)$, where $a_i$ is its $i$-th digit.
Moreover, define $M_{a,i}$ to be the $r$ indices in $[0,r^m-1]$ that differ from $a$ in at most their $i$-th digit. For example, when $r=3,m=4$, we have $e_5=e_{(0,0,1,2)}$, and $M_{5,3}=\{(0,0,0,2)=2,(0,0,1,2)=5,(0,0,2,2)=8\}.$
Next we define $(r+1)m$ subspaces for $i \in [m], u \in [0,r]$:
\begin{eqnarray}
P_{i,u} &=& \spun(e_a: a_i=u), \text{ for $u=0,...,r-1$,} \nonumber \\
P_{i,r} &=& \spun(\sum_{a'\in M_{a,i}}e_{a'}: a \in [0,r^m-1]). \label{eq29}
\end{eqnarray}
Note that for $u\neq r, P_{i,u}$ is spanned by the set of basis vectors whose $i$-th digit index is $u$, and therefore its has dimension $l/r$. It easy to check that also $P_{i,r}$ is a subspace of dimension $l/r$. For example, when $r=3,m=2$,
\begin{align*}
P_{1,0}&=\spun(e_{(0,0)},e_{(0,1)},e_{(0,2)})=\spun(e_0,e_1,e_2),\\
P_{1,1}&=\spun(e_3,e_4,e_5), P_{1,2}=\spun(e_6,e_7,e_8), \text{ and } \\
P_{1,3}&=\spun(e_0+e_3+e_6, e_1+e_4+e_7, e_2+e_5+e_8).
\end{align*}
Using these $k=(r+1)m$ subspaces, we define the $k$ matrices $A_i$ that correspond to the $k$ systematic nodes.

\begin{cnstr} \label{cnstr2}
Let $u \in [0,r],i \in [m]$. For each $um+i\in [k]$, define the matrix $A_{um+i}$ as follows: Its  eigenspaces are $P_{i,u'}, u'\neq u$ that correspond to distinct nonzero eigenvalues. Furthermore, Let $P_{i,u}$ be the repairing subspace, namely $S_{um+i}=P_{i,u}$.
\end{cnstr}

\begin{xmpl}
Deleting node N4 of the code in  Figure \ref{fig1} yields to a $(5,3,2)$ code constructed using Construction \ref{cnstr2}.
Moreover, the code in Figure \ref{fig2} is an $(8,6,4)$ code, constructed using Construction \ref{cnstr2}.
One can check the subspace property holds. For instance, $S_1=\spun\{e_0,e_1\}=\spun\{e_0+e_1,e_1\}$ is an invariant subspace of $A_2$. So $S_1=S_1A_2$. If the two eigenvalues of $A_i$ are distinct, it is easy to show that $S_i \oplus S_i A_i = \mathbb{F}^4$, $\forall i \in [6]$.
\end{xmpl}

\begin{figure*}
  \centering
  \begin{tabular}{|c|c|c|c|c|c|c|}
    \hline
    Node index $i$ & 1 & 2 & 3 & 4 & 5 & 6 \\
    \hline
    Basis for 1st  & $e_0+e_2$ & $e_0+e_1$ & $e_0$ & $e_0$ & $e_0$ & $e_0$ \\
    eigenspace of $A_i$       & $e_1+e_3$ & $e_2+e_3$ & $e_1$ & $e_2$ & $e_1$ & $e_2$ \\
    \hline
    Basis for 2nd & $e_2$ & $e_1$ & $e_0+e_2$ & $e_0+e_1$ & $e_2$ & $e_1$ \\
    eigenspace  of $A_i$       & $e_3$ & $e_3$ & $e_1+e_3$ & $e_2+e_3$ & $e_3$ & $e_3$ \\
    \hline
   Basis for repairing                & $e_0$ & $e_0$ & $e_2$ & $e_1$ & $e_0+e_2$ & $e_0+e_1$ \\
    subspace $S_i$          & $e_1$ & $e_2$ & $e_3$ & $e_3$ & $e_1+e_3$ & $e_2+e_3$ \\
    \hline
  \end{tabular}
  \caption{(n=8,k=6,l=4) code. The first parity node is assumed to be the row sum, and the second parity is computed using encoding matrices $A_i$. Each encoding matrix is defined by its two eigenspaces of dimension $2$. In order to repair node $i$, each surviving node projects its information on the repairing subspace $S_i$, namely it multiplies its columns by the matrix $S_i$. E.g., node $N5$ has two distinct eigenspaces $\spun(e_0,e_1), \spun(e_2,e_3)$. Furthermore, if this node is lost, each surviving node projects its information on the subspace $S_5=\spun(e_0+e_2,e_1+e_3)$.}\label{fig2}
\end{figure*}

\begin{xmpl} \label{xmpl2}
Figure \ref{fig3} illustrates the subspaces $P_{i,u}$ for $r=3$ parities and column length $l=9$. Figure \ref{fig4} is a code constructed from these subspaces with $8$ systematic nodes. One can see that if a node is erased, one can transmit only a subspace of dimension $3$ to repair, which corresponds to only $1/3$ repair bandwidth fraction. Recall that the three encoding matrices for systematic node $i$ are $I, A_i, A_i^2$, for $i \in [8]$.
\end{xmpl}

\begin{figure}
\centering
\begin{tabular}{|c|c|c|c|c|c|c|c|c|}
    															\hline
    															& $P_{1,0}$ & $P_{1,1}$ & $P_{1,2}$ & $P_{1,3}$ & $P_{2,0}$ & $P_{2,1}$ & $P_{2,2}$ & $P_{2,3}$ \\
    															\hline
      														& $e_0$ & $e_3$ & $e_6$ & $e_0+e_3+e_6$ & $e_0$ & $e_1$ & $e_2$ & $e_0+e_1+e_2$ \\
    \text{Basis for the subspace} & $e_1$ & $e_4$ & $e_7$ & $e_1+e_4+e_7$ & $e_3$ & $e_4$ & $e_5$ & $e_3+e_4+e_5$ \\
      														& $e_2$ & $e_5$ & $e_8$ & $e_2+e_5+e_8$ & $e_6$ & $e_7$ & $e_8$ & $e_6+e_7+e_8$ \\
    															\hline
\end{tabular}
\caption{Basis Sets of vectors used to construct a code with $r=3$ parities and column length $l=3^2=9$.}
\label{fig3}
\end{figure}

\begin{figure}
\centering
\begin{tabular} {|c|c|c|c|c|c|c|c|c|}
    \hline
    $i$ & 1 & 2 & 3 & 4 & 5 & 6 & 7 & 8 \\
    \hline
    													 & $P_{1,3}$ & $P_{2,3}$ & $P_{1,0}$ & $P_{2,0}$ & $P_{1,0}$ & $P_{2,0}$ & $P_{1,0}$ & $P_{2,0}$ \\

    \text{The $3$ eigenspaces} & $P_{1,1}$ & $P_{2,1}$ & $P_{1,3}$ & $P_{2,3}$ & $P_{1,1}$ & $P_{2,1}$ & $P_{1,1}$ & $P_{2,1}$ \\

    													 & $P_{1,2}$ & $P_{2,2}$ & $P_{1,2}$ & $P_{2,2}$ & $P_{1,3}$ & $P_{2,3}$ & $P_{1,2}$ & $P_{2,2}$ \\
    \hline
    Repairing subspace 				 & $P_{1,0}$ & $P_{2,0}$ & $P_{1,1}$ & $P_{2,1}$ & $P_{1,2}$ & $P_{2,2}$ & $P_{1,3}$ & $P_{2,3}$ \\
    \hline
\end{tabular}
\caption{An $(n=11,k=8,l=9)$ code. The subspaces $P_{i,u}$ are listed in Figure \ref{fig3}.}
\label{fig4}
\end{figure}

The following theorem shows that the code indeed has optimal repair bandwidth $1/r$.
\begin{thm}\label{thm2}
Construction \ref{cnstr2} has optimal repair bandwidth $1/r$ when repairing one systematic node.
\end{thm}
\begin{IEEEproof}
For distinct integers $um+i,u'm+i'\in [k]$ for $u,u'\in [0,r-1]$ and $i,i'\in [m]$ we will show that \eqref{eq24} is satisfied, namely
$$S_{um+i}A_{u'm+i'}=S_{um+i}.$$

\begin{itemize}
\item Case $i\neq i'$:  It is easy to verify that the $r$ eigenspaces $T_1,...,T_r$ of  $A_{u'm+i'}$ satisfy
	\begin{equation}
    P_{i,u}=\sum_{j=1}^r(P_{i,u}\cap T_j).
    \label{eq:145}
    \end{equation}
    Notice that \eqref{eq:145} is usually \emph{not} correct for arbitrary subspaces $T_1,...,T_r$ that satisfy
    $\sum_{i}T_i=\mathbb{F}^l$. By definition $S_{um+i}=P_{i,u}$, then
	\begin{align*}
	S_{um+i}A_{u'm+i'}&=P_{i,u}A_{u'm+i'}\\
	&=(\sum_{j=1}^rP_{i,u}\cap T_j)A_{u'm+i'}\\
	&=\sum_{j=1}^r(P_{i,u}\cap T_j)A_{u'm+i'}\\
	&=\sum_{j=1}^r(P_{i,u}\cap T_j)\\
	&=P_{i,u}\\
	&=S_{um+i}.
	\end{align*}	
\item  Case $i=i'$, and  $u\neq u'$: By the construction, the eigenspaces of $A_{u'm+i}$ are $\{P_{i,1},...,P_{i,r}\}\backslash \{P_{i,u'}\}$. Since $u\neq u'$ then $P_{i,u}\in \{P_{i,1},...,P_{i,r}\}\backslash \{P_{i,u'}\}$, and
    $$S_{um+i}A_{u'm+i'}=P_{i,u}A_{u'm+i'}=P_{i,u}=S_{um+i}.$$
\item Case $i=i'$, and  $u=u'$: In this case we will only prove the case where $u=0$. The rest of the cases are proved similarly.
    Denote by $A_{um+i}=A,S=S_{um+i}$, then by \eqref{eq25} we need to show that
    $$S+SA+..+SA^{r-1}=\mathbb{F}^{l}.$$
    Denote the distinct eigenvalues of $A$ by $\lambda_0,\lambda_1,\dots,\lambda_{r-1}$.
    For a vector  $a=(a_1,a_2,\dots,a_m)$ or equivalently an integer $a \in [0,l-1]$, denote by $a_i(u)=(a_1,\dots,a_{i-1},u, a_{i+1},\dots,a_m)$ the vector that is the same as $a$ except the $i$-th entry, which is $u$.
    Notice that $S=\spun(P_{i,0})=\spun\{e_{a_i(0)}:\forall a \in [0,l-1]\}$ and
    \begin{eqnarray*}
    &&e_{a} A^s\\
     &=& (\sum_{u=0}^{r-1} e_{a_i(u)}-e_{a_i(1)}-\dots-e_{a_i(r-1)}) A^s \\
      &=& \lambda_0^s \sum_{u=0}^{r-1} e_{a_i(u)} - \lambda_1^s e_{a_i(1)} - \dots -  \lambda_{r-1}^s e_{a_i(r-1)} \\
      &= & \lambda_0^s e_{a_i(0)} + \sum_{u=1}^{r-1}(\lambda_0^s - \lambda_{u}^s) e_{a_i(u)}.
     \end{eqnarray*}
     Writing the equations for all $s \in [0,r-1]$ in a matrix, we get
     $$\left(
     \begin{array}{c}
       e_{a_i(0)} \\
       e_{a_i(0)} A \\
       e_{a_i(0)} A^2 \\
       \vdots \\
       e_{a_i(0)} A^{r-1} \\
     \end{array}
    \right)
    = M  \left(
    \begin{array}{c}
     e_{a_i(0)} \\
     e_{a_{i(1)}} \\
     \vdots \\
     e_{a_i(r-1)} \\
    \end{array}
    \right),$$
    with
     $$ M=\left(
     \begin{array}{cccc}
       1 & 0 & \cdots & 0 \\
       \lambda_0 & \lambda_{0}-\lambda_{1} & \cdots & \lambda_{0}-\lambda_{r-1} \\
       \lambda_0^2 & \lambda_{0}^2-\lambda_{1}^2 & \cdots & \lambda_{0}^2-\lambda_{r-1}^2 \\
       \vdots & \vdots &  & \vdots \\
       \lambda_0^{r-1} & \lambda_{0}^{r-1}-\lambda_{1}^{r-1} & \cdots & \lambda_{0}^{r-1}-\lambda_{r-1}^{r-1}
     \end{array}
    \right).
    $$
    After a sequence of elementary column operations, $M$ becomes the following Vandermonde matrix
    $$M'=\left(
        \begin{array}{cccc}
          1 & 1 & \cdots & 1 \\
          \lambda_0 & \lambda_1 & \cdots & \lambda_{r-1} \\
          \lambda_0^2 & \lambda_1^2 & \cdots & \lambda_{r-1}^2 \\
          \vdots & \vdots &  & \vdots \\
          \lambda_0^{r-1} & \lambda_1^{r-1} & \cdots & \lambda_{r-1}^{r-1} \\
        \end{array}
      \right).$$
    Since $\lambda_i$'s are distinct, we know $M'$ and hence $M$ is non-singular. Therefore, $\spun\{e_{a_i(0)}, e_{a_i(0)}A,\dots,e_{a_i(0)}A^{r-1}\}$ $= \spun\{e_{a_i(0)},e_{a_i(1)},\dots,e_{a_i(r-1)}\}$. Since $S_i$ contains $e_{a_i(0)}$ for all $r$-ary vector $a$, we know $S_i + S_i A_i + \dots + S_i A_i^{r-1} = \mathbb{F}^l$.
\end{itemize}
\end{IEEEproof}

\subsection{Update and access complexity}
We discuss the update and access complexity of our code in this subsection. First we make some observations. 
\begin{enumerate}
  \item The code restricted to the systematic nodes $i\in [m], u=r$ is equivalent to that of \cite{viveck,Tamo2}. Since the encoding matrices $A_i^Q$, are all diagonal, every information entry appears exactly once in each of the two parities, and therefore it appears $r+1$ times in the code (once in each of the parities and once in its systematic node). Clearly this is the minimum possible, since the code is an MDS. As mentioned in the introduction, this is an \emph{optimal-update} code. In \cite{Tamo3} it was proven that an optimal-update code with diagonal encoding matrices has no more than $m$ systematic nodes. But we will show an optimal-update construction in the next section with $2m$ systematic nodes but non-diagonal encoding matrices.
  \item Shortening the code to contain only the $rm$ systematic nodes $i\in [m],u \in [0,r-1]$ will result a code $\cC$ that is actually equivalent to the code in \cite{viveck3}. We assume here that $\{e_a, a \in [0,l-1]\}$ are standard basis. Namely, each repairing subspace $P_{i,u}$ can be represented by an $l/r\times l$ matrix, such that each row has exactly one nonzero entry. Therefore when repairing a node, only $l/2$ symbols from each surviving node are being read and transmitted to the repair center, with no need of any computations within the surviving node (e.g. Figure \ref{fig2}). Such a code is termed to have \emph{optimal access}. It was shown in \cite{Tamo3} that a code with optimal access has at most $2m$ nodes, therefore this construction is optimal. Namely it is a code with optimal access and maximum possible number of systematic nodes.
  \item We conclude that the code construction is a combination of the longest optimal-access code and the longest optimal-update code (with diagonal encoding matrices), which provides an interesting tradeoff among access, update, and the code length. In other words, we can achieve a larger number of nodes if we are willing to sacrifice the optimal-access and/or optimal-update properties. The shortening technique was also used in \cite{Kumar2009}\cite{Suh-alignment} in order to get optimal-repair code with different code rates.
\end{enumerate}

Clearly, the optimal-access property is highly desirable in a code. Therefore one might ask what is the longest code (in terms of $k$), that has the maximum number of nodes that can be repaired with optimal access. In particular let us consider codes with 2 parities. If we try to extend the optimal-access code $\cC$ with $2m$ systematic nodes to an optimal repair code $\cD$ with $k$ systematic nodes, then $k\leq 3m$, as the following theorem suggests. Therefore, our construction is longest in the sense of extending $\cC$.
Before proving the theorem we will need the following lemma.
\begin{lem}\cite[Lemma 8]{Tamo3}
The repairing subspaces $S_i$ of an optimal repair $(k+2,k,l)$ code satisfy that for any subset of indices $J\subseteq [k]$
$$\dim(\cap_{i\in J}S_i)\leq \frac{l}{2^{|J|}}.$$
\label{good lemma}
\end{lem}
\begin{thm} \label{thm3}
Any extension of an optimal access code with $2m$ systematic nodes to an optimal repair code, will have no more than $3m$ systematic nodes, for $r=2$ parities.
\end{thm}

\begin{IEEEproof}
Let $\cC$ be an optimal-access code of length $2m$ with 2 parities. Let $\cD$ be an extended code of $\cC$.
By equivalently transforming the encoding matrices (see \cite{Tamo3}), we can always assume the encoding matrices of the parities in $\cD$ are
$$\left(
    \begin{array}{cccccc}
      I & \cdots & I & I & \cdots & I \\
      A_1 & \cdots & A_{2m} & A_{2m+1} & \cdots & A_{k} \\
    \end{array}
  \right).
$$

Here the first $2m$ column blocks correspond to the encoding matrices of $\cC$.
First consider the code $\cC$, that is the first $2m$ nodes.
If $\cC$ has optimal access, then each repairing subspace is spanned by $l/2$ standard basis vectors.
Since $\cC$ contains $2m$ systematic nodes, on average each standard basis vector appears in $2m\times \frac{l}{2} \times \frac{1}{l} = m$ repairing subspaces.
For each $i=0,...,l-1$ let $J\subseteq [2m]$ be the subset of indices of the repairing subspaces that contain the vector $e_i$. We claim that each standard basis vector appears exactly $m$ times, namely for each $i$ the size of $J$ is $m$. Assume to the contrary that $|J|>m$ for some $i$. By Lemma \ref{good lemma}
$$1\leq \dim(\cap_{i \in J} S_i) \le \frac{2^m}{2^{|J|}}<1,$$
and we get a contradiction. Moreover, if there exists $J$ of size less than $m$, then by a simple counting argument we get that there exists an $J'$ of size greater than $m$, which can not happen. Hence, we conclude that for each $i$ the size of $J$ is exactly $m$ and,
$$\spun(e_i)=\cap_{i\in J}S_i.$$
Now consider a systematic node $j\in [2m+1,k]$ that was added to the code $\cC$. Since $\cD$ is an optimal repair code, each repairing subspace of the nodes in $\cC$ is an invariant subspace of $A_j$. Since the intersection of invariant subspaces is again an invariant subspace we get that for any $i=0,...,l-1$
$$\cap_{i\in J}S_i=\spun(e_i)$$
is an invariant subspace of $A_j$. Namely, each standard basis vector is an eigenvector of $A_j$, and therefore $A_j$ is a diagonal matrix. We conclude that restricting the code $\cD$ to its last $k-2m$ systematic nodes will yield to an optimal update code. By \cite{Tamo3}[Theorem $6$], there are only $m$ nodes that are all optimal update, hence $k-2m\le m$.
\end{IEEEproof}

\subsection{Reconstruction and finite field size}
Next we will show that the code can be made to be MDS over a large finite field.
\begin{thm}\label{thm5}
The code can be made an MDS over a field large enough.
\end{thm}
\begin{IEEEproof}
Assign arbitrarily $r$ distinct nonzero eigenvalues to each matrix $A_i$. Recall that the encoding matrices are defined as $A_{s,i}=A_i^{s-1}$, therefore each one of them is invertible. We multiply each encoding matrix $A_{s,i}$ by a specific variable $x_{(s-1)k+i}$, to get a new code defined by the matrix  \begin{equation} \label{eq17}
\left[ \begin{array}{ccc}
x_{1}A_{1,1}  & \cdots & x_{k}A_{1,k} \\
\vdots & \ddots & \vdots \\
x_{(r-1)k+1}A_{r,1}  &\cdots & x_{rk}A_{r,k}
\end{array} \right].
\end{equation}
Clearly the new code is MDS iff any $t\times t$ block submatrix in \eqref{eq17} is invertible, for any $t\in [r]$. Define the multivariate polynomial $P$ in the variables $x_{s,i}$, which is the product of the determinants of all the $t\times t$ block submatrices, for any $t=1,...,r$. Hence, the code can be made to be MDS if there is an assignment to the variables that does not evaluate $P$ to zero. Let $x=(x_{1},...,x_{rk})$ be  the vector of the variables. For a vector of integers $a=(a_1,...,a_{rk})$ we define $x^a=\prod_{i}x_{i}^{a_i}$. Furthermore, define the usual ordering on the terms $x^a$ according to the lexicographic order, i.e., $x^a\ge x^b$ iff $a\ge b$ according to the lexicographic order. The leading coefficient of a multivariate polynomial, is the coefficient of the maximal nonzero term. For example, the leading coefficient of the polynomial $2x_1^2x_3+x_1^2x_4$ is $2$.

Let $a=\{a_1<a_2<...<a_t\}$ and $b=\{b_1<b_2<...<b_t\}$ be two sets of indices of size $t$ in $[r]$ and $[k]$ respectively. Define $P_{a,b}$ to be the determinant of the submatrix restricted to row blocks $a$ and column blocks $b$. It is easy to see that its leading coefficient is $$\prod_{i}^t\det(A_{a_i,b_i}),$$  which is non zero, since by construction, each of matrices is invertible. Moreover if $P_1,P_2$ are the determinant of different submatrices, then the leading coefficient of their product $P_1\cdot P_2$, is the product of their leading coefficients. Since both of them are non zero, so is the product.
$P$ is a product of such polynomials $P_i$, therefore also its leading coefficient is non zero. Moreover, each $P_i$ is an homogeneous polynomial, hence so is $P$.  We conclude that $P$ has a nonzero term $x^a$ (its leading coefficient) of degree equal to $\deg(P)$. By the Combinatorial Nullstellensatz \cite{Alon-polynomial-method} we get that a field of size greater than  $\max_i\{a_i:a=(a_1,...,a_{rk})\}$ will suffice.
\end{IEEEproof}

For the case of 2 parities, we can explicitly specify the finite field size.
The following construction defines uniquely the encoding matrices, by defining their eigenvalues. This assignment of the eigenvalues guarantees the MDS property of the optimal repair code.
\begin{cnstr}
\label{construction 2}
Let $\{\lambda_{i,j}\}_{i\in [m],j=0,1}$ be an arbitrary $2m$ distinct non zero elements of the field $\mathbb{F}_q$, $q \ge 2m+1$.
Assign arbitrarily to each eigenspace of the matrix $A_{um+i}$, the eigenvalue $\lambda_{i,0}$ or $\lambda_{i,1}$, as long as
each $P_{i,u'}$ correspond to distinct eigenvalues in the two matrices it appears as an eigenspace, $u,u' \in \{0,1,2\}$.
\end{cnstr}

For example, we can assign eigenvalues in the following way:
$$
\begin{array}{c|cccc}
  \textrm{encoding matrix} & \textrm{1st eigensapce} & \textrm{1st eigenvalue} & \textrm{2nd eignenspace} & \textrm{2nd eignvalue} \\
  \hline
  A_{i}    & P_{i,2} & \lambda_{i,1} & P_{i,1} & \lambda_{i,0}  \\
  A_{m+i}  & P_{i,0} & \lambda_{i,1} & P_{i,2} & \lambda_{i,0}  \\
  A_{2m+i} & P_{i,0} & \lambda_{i,0} & P_{i,1} & \lambda_{i,1} \\
\end{array}
$$
Take the case of $m=2$ in Figure \ref{fig2}, we can use finite field $\mathbb{F}_5$ and assign the eigenvalues to be
$$(\lambda_{1,0},\dots,\lambda_{6,0})=( 1,2,1,2,4,3),$$
$$(\lambda_{1,1},\dots,\lambda_{6,1})=( 4,3,4,3,1,2).$$

Remark: If we have an extra systematic column with $A_{3m+1}=I$ (see column $N4$ in Figure \ref{fig1}), we can use a field of size $2m+2$ and simply modify the above construction such that all $\lambda_{i,j} \neq 1$, for $i \in [3m], j=0,1$. For example, when $m=1$, the coefficients in Figure \ref{fig1} are assigned using the above algorithm, where the field size is $4$.

\begin{thm}
There is an optimal repair $(3m+2,3m,2^m)$ MDS code if the finite field size is at least $2m+1$.
\end{thm}
\begin{IEEEproof}
We will show that Construction \ref{construction 2} satisfies the MDS property, namely, any two erasures can be repaired.
This is equivalent to that (i) all the encoding matrices $A_x$'s are invertible, and (ii) any $2 \times 2$ block sub matrix
$$
\left[ \begin{array}{cc}
I & I \\
A_x & A_y
\end{array} \right]
$$
is invertible, for any distinct $x,y \in [k]$.
Since the eigenvalues are nonzero the first condition is satisfied.
The second condition is equivalent to that $A_x-A_y$ is invertible. Let $x=um+i,y=vm+j$, with $i,j \in [u], u,v \in \{0,1,2\}$.
\begin{itemize}
\item Case $i\neq j$: Let the eigenspaces of $A_{um+i},A_{vm+j}$ be  $V_1,V_2$ and $U_1,U_2$ respectively, which correspond to eigenvalues
    $\lambda_1,\lambda_2$, and $\mu_1,\mu_2$.
	Clearly $$V_1\oplus V_2=U_1\oplus U_2=\mathbb{F}^l.$$ It is easy to check that
	$$\oplus_{i,j=1}^2(V_i\cap U_j)=\mathbb{F}^l.$$ Assume to the contrary that there exists a non zero vector $a$ such that $$a(A_{um+i}-A_{vm+j})=0,$$
    where $a=\sum_{i,j=1}^2a_{i,j}, \text{ and } a_{i,j}\in V_i\cap U_j.$
	Then,
	 $$0=a(A_{um+i}-A_{vm+j})=(\lambda_1-\mu_1)a_{1,1}+(\lambda_1-\mu_2)a_{1,2}+(\lambda_2-\mu_1)a_{2,1}+(\lambda_2-\mu_2)a_{2,2}.$$
	 Since $a$ is non zero, at least one of the $a_{i,j}$'s is non zero. Hence, $\lambda_i=\mu_j$ and we get a contradiction since the eigenvalues of
    $A_{um+i}$, and $A_{vm+j}$ are distinct.
\item Case $i=j$ and $u\neq v$: Since $i=j$ the matrices $A_{um+i}$ and $A_{vm+i}$ share a common eigenspace from the set of subspaces $\{P_{i,u}, u \in \{0,1,2\}\}$. Denote by $V,U$ and $V,W$ the eigenspaces of $A_{um+i}$, $A_{vm+i}$. Denote by $\lambda, \mu$ the eigenvalues that correspond to the eigenspace of $V$ in the matrices $A_{um+i},A_{vm+i}$ respectively. By construction, $\lambda\neq \mu$, and therefore by construction $U$ is an eigenspace of $A_{um+i}$ with an eigenvalue $\mu$, and $W$ is an eigenspace of $A_{vm+i}$ with an eigenvalue $\lambda$.
    Assume that $aA_{um+i}=aA_{vm+i}$ for some non zero vector
    \begin{equation}
    a=b+c=b'+d,
    \label{blibli}
    \end{equation}
    where $b,b'\in V$, $c\in U$, and $d\in W.$ Then
	$$\lambda b+\mu c=(b+c)A_{um+i}=aA_{um+i}=aA_{vm+i}=(b'+d)A_{vm+i}=\mu b'+\lambda d,$$
	using \eqref{blibli} we conclude that $\mu=\lambda$ which is a contradiction.
\end{itemize}
\end{IEEEproof}

One can observe that the proposed code construction has parameters $(3m+2,3m,2^m)$, and a field size that scales linearly with the number of systematic nodes. On the other hand, the $(m+3,m+1,2^m)$ code in \cite{Tamo2} requires only a field of size $3$.
Thus, the proposed code can protect more systematic nodes, but has longer (actual) column length. The actual size of each column is longer since it has to store $2^m$ symbols of a \emph{larger} field. Nonetheless, it may be possible to alter the structure of the encoding matrices a bit (for example, relaxing the requirement that each of the encoding matrix is diagonalizable), and obtain a constant field size. This remains as a future research direction.

%
\section{Long Optimal-Update Code} \label{sec7}
%
%
In storage systems that use coding to combat failures, each parity symbol is a function of a subset of information symbols. Therefore, when an information symbol updates its value, also the  parity symbols that are function of it, need to be updated. Since update is one of the most frequent operation in the system, one would like to minimize the amount of symbols' update incurred by one information symbol update. In an MDS code each parity node is a function of the \emph{entire} information symbols, hence at least one parity symbol needs to be updated in any information symbol update. An optimal update MDS code attains this lower bound, namely each parity node updates exactly one of its symbols for each information symbol update. It is easy to see that in an optimal update \emph{linear} code, each encoding matrix is a generalized permutation matrix, i.e. there is exactly one nonzero entry in each row and each column.

In \cite{Tamo3} diagonal encoding matrices ,which are a special case of generalized permutation matrices, were considered.  They showed that an optimal bandwidth MDS code with $2$ parities, and diagonal encoding matrices, has at most $\log_2 l$ systematic nodes. In this section we will show that one can improve that by not restricting to diagonal encoding matrices. More precisely, we will construct on optimal update code with $2\log_2 l$ systematic nodes.

%
%
%

Let $l=2^m$ for some integer $m$, and define for any $i=1,...,m$ the following four subspaces of $\mathbb{F}^l$ of dimension $l/2$:
\begin{eqnarray*}
  P_i &=& \spun(e_a: a_{i}=0), \\
  R_i &=& \spun(e_a: a_{i}=1), \\
  Q_i &=& \spun(ye_a+x e_b: a_{i}=0, b_{i}=1,a_{j}=b_{j}, \forall j \neq i),\\
  O_i &=& \spun(-ye_a+xe_b:a_{i}=0, b_{i}=1, a_{j}=b_{j}, \forall j \neq i),
\end{eqnarray*}
where $x$ and $y$ are non zero elements of the field that satisfy  $x^2 \neq y^2.$ In the following we will also use letters $P,Q$ as superscripts for the encoding matrices.

\begin{cnstr} \label{cnstr3}
Construct the $(n=2m+2,k=2m,l=2^m)$ code over $\mathbb{F}$ by the following $2m$ encoding matrices $A_i^T,i=1,...,m$ and $T=P,Q$.
\begin{itemize}
	\item Define the matrix $A_i^P$ to have eigenspaces $Q_i, O_i$ that correspond to eigenvalues $xy, -xy$ respectively.
	\item Define the matrix $A_i^Q$ to have eigenspaces  $P_i, R_i$ that correspond to distinct non zero eigenvalues $\lambda,\mu$ respectively.
\end{itemize}
Moreover,  let the repairing subspace that correspond to the matrix $A_i^T$ be $S_i^T=T_i$.
\end{cnstr}

E.g., when  $m=1$, we get a $(4,2,2)$ with $2$ encoding matrices represented with respect to the standard basis
\begin{equation} 
  A_1^Q=\left[\begin{array}{cc}
\lambda & \\
  & \mu \end{array} \right],
  A_1^P \left[\begin{array}{cc}
 & x^2\\
y^2 &  \end{array} \right].
\end{equation}
and repairing subspaces
$$S_1^Q =Q_1= (y,x),  S_1^P=P_1= (1,0).
$$
When $m=2$, the encoding matrices are
$$
A_1^Q=\left[\begin{array}{cccc}
\lambda & & & \\
  & \lambda & &\\
  & & \mu & \\
  & & & \mu \end{array} \right],
A_2^Q=\left[\begin{array}{cccc}
\lambda & & & \\
  & \mu & &\\
  & & \lambda & \\
  & & & \mu \end{array} \right],
A_1^P=\left[\begin{array}{cccc}
 & & x^2 & \\
  &  & & x^2\\
y^2  & &  & \\
  & y^2 & &  \end{array} \right],
A_2^P=\left[\begin{array}{cccc}
 & x^2& & \\
 y^2 &  & &\\
  & &  & x^2\\
  & & y^2 &  \end{array} \right].
$$
The repairing subspaces are
$$
\left[\begin{array}{cccc}
y & & x & \\
  & y & & x \end{array} \right],
\left[\begin{array}{cccc}
y & x &  & \\
  &  & y & x \end{array} \right],
\left[\begin{array}{cccc}
1 & &  0&0 \\
  & 1 & 0&0  \end{array} \right],
\left[\begin{array}{cccc}
1 &0 &  & \\
  &  &1 & 0 \end{array} \right].
$$
In both cases it is not difficult to check that the subspace property is satisfied, hence the code has optimal bandwidth. And since the encoding matrices are permutation matrices, the code has optimal update.

\begin{thm}
  Construction \ref{cnstr3} has optimal bandwidth and optimal update.
\end{thm}
\begin{proof}
  It is easy to see that the encoding matrices are all permutation matrices, so the code has optimal update.
  We need to show the subspace property, namely
  for $i,j\in [m]$ and $Y,T\in \{P,Q\}$
  \begin{equation*}
S_i^YA_j^T\cap S_i^Y=\begin{cases}
\{0\} &i=j \text{ and } Y=T\\
S_i^Y & \text{ otherwise.}
\end{cases}
\end{equation*}

  \begin{itemize}
    \item Case $i \neq j$: One can check that for $Y \in \{Q,P\}$,
    $$Y_i=(Y_i\cap P_j)\oplus (Y_i\cap R_j)\text{ and } Y_i=(Y_i\cap Q_j)\oplus (Y_i\cap O_j).$$
    Therefore the  proof is the same as  in Theorem \ref{thm2}.
    \item Case $i=j,\text{ and } Y\neq T$: In this case $Y_i$ is an eigenspace of $A_i^T$, and the result follows.
    \item Case $i=j,T=Y$: Assume that $Y=P$, and we will show that the transformation $A_i^P$ maps the subspace $S_i^P=P_i$ to the subspace $R_i$, and since $P_i\cap R_i=\{0\}$ the result will follow.
       let $e_a \in P_i$  and $b$ be the integer that is identical to $a$ except on the $i$-th digit. Then
    \begin{eqnarray*}
      e_a A_i^P &=& \frac{1}{2y}[(ye_a + xe_b)-(-ye_a+xe_b)]A_i^P \\
      &=& \frac{xy}{2y} (ye_a + xe_b)- \frac{-xy}{2y}(-ye_a+xe_b) \\
      &=& x^2 e_b\in R_i.
    \end{eqnarray*}
    When $Y=Q$ the result follows by the same reasoning.
  \end{itemize}
\end{proof}

Similar to Theorem \ref{thm5} it is clear that the code can be MDS over a large enough finite field. To summarize the result of this section, we gave a construction that doubled the number of systematic nodes compared to the bound in \cite{Tamo3}. The reason for the violation of this bound is by not restricting to diagonal encoding matrices.
%
\section{Lowering the Access Ratio} \label{sec6}
%
%
Repairing a failed node is a computationally heavy task, that requires large amount of the system's resources. Therefore, optimizing the repair algorithm is of high importance. One way to optimize is by reducing the amount of symbols needed to be accessed and read during the repair process. This parameter is quantified by the \emph{access ratio} of the system. In this section we will use  explicit linear transformations performed on the code in Construction \ref{cnstr2} that yields to an equivalent code with a lower access ratio during a repair process.
Furthermore, these transformations maintain the other properties of the code, namely the MDS and the optimal repair properties.

Formally, given an $(n,k,l)$ code $\mathcal{C}$, let $\beta(i)$ denote the number of symbols (or entries) accessed in the surviving nodes during the repair of systematic node $i$. The \emph{access ratio} is defined as
$$R = \frac{\sum_{i=1}^{k} \beta(i)}{k(n-1)l}.$$
Note that $(n-1)l$ is the amount of surviving symbols in the system in the event of one node erasure, hence $R$ is the average fraction of the number of symbols in the system being accessed  during a repair process.
The $((r+1)m+r,(r+1)m,r^m)$ code in Construction \ref{cnstr2} has $(r+1)m$ systematic nodes, where $rm$ of them are repaired with optimal access, i.e., only $l/r$ symbols are accessed from each node during the repair process. Thus, repairing these nodes costs accessing $rm \cdot (n-1)l/r$ symbols. However, repairing any of the rest $m$ systematic nodes, one has to access \emph{all} the surviving symbols in the system. Notice that, although the repair is optimal, in order to generate the transmitted data one has to access the entire information in the node. Repairing these nodes costs accessing $m \cdot (n-1)l$ symbols, and the access ratio of the code is
\begin{equation}\label{eq34}
R=\frac{rm \cdot (n-1)l/r+m \cdot (n-1)l}{(r+1)m \cdot (n-1)l} = \frac{2}{r+1}.
\end{equation}
This value of the access ratio $R=2/(r+1)$ is our benchmark. We will show that with an appropriate selection of linear transformation, the value of access ratio $R$ can be reduced. But first we define how to apply linear transformation on the code to receive an equivalent code. Moreover  we will show that these linear transformations preserve the ``nice'' properties of our code.

Let $A=(A_{i,j})_{i\in [r],j\in [k]}$ be the encoding matrix of an $(k+r,k,l)$ optimal repair MDS code, with repairing subspaces $S_i,i=1,...,k.$
We will apply a linear transformation on the code by multiplying on the right the encoding matrix $A$  by a block diagonal matrix $B$, to get the encoding matrix $C$ as follows,
\begin{eqnarray*}
C=\left[ \begin{array}{ccc}
C_{2,1}  & \cdots & C_{2,k} \\
\vdots & \ddots & \vdots \\
C_{r,1} &\cdots & C_{r,k}
\end{array} \right]= AB
=
\left[ \begin{array}{ccc}
A_{2,1}  & \cdots & A_{2,k} \\
\vdots & \ddots &  \vdots\\
A_{r,1} &\cdots & A_{r,k}
\end{array} \right]
\left[
\begin{array}{ccc}
B_1 & & \\
& \ddots & \\
& & B_k
\end{array}
\right].
\end{eqnarray*}
Namely, for $i \in [r], j \in [k]$
\begin{equation} \label{eq33}
C_{i,j}=A_{i,j}B_{j},
\end{equation}
where $B_j$ is an invertible matrix of size $l \times l$.
After applying the linear transforation $B$ on the encoding matrix, the repairing subspaces should be changed accordingly. Recall that $S_{i,j}$ is the repairing subspace for surviving node $j$ during the repair of node $i$. Define the new repairing subspaces as follows:
\begin{equation} \label{eq32}
  S_{i,j}= \begin{cases}
  S_i B_j, & j \in [k],\\
  S_i, & j \in [k+1,k+r].
\end{cases}
\end{equation}
Notice that compared to the original code, the repairing subspaces are changed only for the systematic nodes.

\begin{thm}
Consider the linear transformation defined by \eqref{eq33}\eqref{eq32} applied on an optimal-bandwidth MDS code, then the resulting code is an optimal-bandwidth MDS code, with repairing subspaces $S_{i,j}$.
\end{thm}
\begin{IEEEproof}
Since the code defined by the encoding matrix $A$ is optimal bandwidth, then by the subspace property \eqref{eq8}\eqref{eq9}  for any distinct $i, j \in [k]$, and $ t \in [r]$,
$$S_i = S_i A_{t,j}.$$
Therefore,
$$S_{i,j}=S_i B_j =  S_i A_{t,j}B_j=S_{i,k+t} C_{t,j}.$$ And
\eqref{eq6} is satisfied. Moreover, the sum of subspaces satisfies
$$\sum_{t=1}^{r}S_i A_{t,i}=\mathbb{F}^l,$$
therefore
$$\sum_{t=1}^rS_{i,k+t} C_{t,i}=\sum_{t=1}^rS_i A_{t,i}B_i=\mathbb{F}^l.$$
Therefore \eqref{eq7} is satisfied, and the equivalent code $C$ has optimal bandwidth.
It is easy to check that if $A$ is an MDS code, then also $C$, and the result follows.
\end{IEEEproof}

Now let us find a code such that the number of accesses will be decreased.
We say node $j$ has optimal access during the repair of node $i$, if only $l/r$ symbols are to be accessed in node $j$ during the repair on node $i$. This is equivalent to the following \textbf{optimal-access condition:} $S_{i,j}=S_i B_j$ can be written as a matrix with only $l/r$ non-zero columns.
So we need to look for proper $B_j$'s such that this condition is satisfied by as many pairs $(i,j)$ as possible.
Let $V_j$ be the matrix of the left eigenspaces of the encoding matrix $A_j$ in Construction \ref{cnstr2}, and we call it \emph{eigenspace matrix}. When $j=vm+y$, for $v \in [0,r], y \in [m]$, we have
$$V_j =\begin{pmatrix}
  P_{y,0} \\
  \vdots \\
  P_{y,v-1} \\
  P_{y,v+1} \\
  \vdots \\
  P_{y,r}
\end{pmatrix},  $$
where $P_{y,u'}$ are defined as in \eqref{eq29}. Here we view each $P_{y,u'}$ as $l/r$ of vectors instead of a subspace. For example, for the code in Figure \ref{fig2} if $j=1$ and consider standard basis $\{e_0,e_1,e_2,e_3\}$ then
$$V_1 = \begin{pmatrix}
  e_0 + e_2 \\
  e_1 + e_3 \\
  e_2 \\
  e_3 \\
\end{pmatrix}=\begin{pmatrix}
  1 & 0 & 1 & 0 \\
  0 & 1 & 0 & 1 \\
  0 & 0 & 1 & 0 \\
  0 & 0 & 0 & 1 \\
\end{pmatrix}.$$
Define the matrix of transformation as
\begin{equation} \label{eq20}
B_j = V_j ^{-1},
\end{equation}
which is the inverse of the eigenspace matrix.

\begin{thm}
The access ratio of the $(n=(r+1)m+r,k=(r+1)m,l=r^m)$ code using \eqref{eq20} is $$\frac{2}{r+1}-\frac{r-1}{(n-1)(r+1)}.$$
\end{thm}
\begin{IEEEproof}
Suppose node $i=um+x$ is erased. From node $j=vm+y$, $j \neq i$, by \eqref{eq32} we need to send the following subspace:
$$S_{i,j} = S_{i} B_{j} =  S_i V_{j} ^{-1}.$$
Here $S_{i}$ is defined as $P_{x,u}$ as in Construction \ref{cnstr2}, and $B_{j}$ is defined in \eqref{eq20}. We are going to show that in a lot of cases $S_i$ can be rewritten as the product of a matrix $M$ and the eigenspace matrix $V_j$:
\begin{equation}\label{eq30}
S_i = M V_j,
\end{equation}
where $M$ is of size $l/r \times r$ and contains only $l/r$ non-zero columns. This will lead to $S_{i,j}=M V_j V_j^{-1} = M$ and therefore the code will have optimal access for the pair $i,j$.
\begin{itemize}
  \item Case $x=y$, $u \neq v$. Apparently, $S_i = P_{x,u}$ is one of the eigenspaces in $V_j$ and \eqref{eq30} is satisfied.
  \item Case $x \neq y$, $u \neq r$. We have observed in \eqref{eq:145} that the subspaces satisfy $P_{x,u}=\sum_{j=1}^r(P_{x,u}\cap T_j)$, where $T_1,\dots,T_r$ are all the eigenspaces of $A_j$. Moreover, each $P_{x,u}\cap T_j$ only contains linear combinations of $l/r^2$ vectors in $T_j$. Hence \eqref{eq30} holds.
  \item Case $x \neq y$, $u = r$. We need to access all remaining elements.
\end{itemize}
Recall the code length is $k=(r+1)m$. Hence for each systematic node $i$ as a survived node, it has optimal access for $r+(m-1)r = mr$ erased nodes (the first two cases), and accesses all elements for $m-1$ erased nodes (the last case). For each parity node as a survived node, it has optimal access for $rm$ erased nodes ($j \in [rm]$), and accesses all elements for $m$ erased nodes ($j \in [rm+1,(r+1)m]$), because the repairing subspaces are still $S_i$ for parity nodes.
Therefore, the access ratio is
$$\frac{k(rm\frac{l}{r}+(m-1)l)+r(rm\frac{l}{r}+ml)}{k(n-1)l}=\frac{2}{r+1}-\frac{r-1}{(n-1)(r+1)}. $$
Hence the proof is completed.
\end{IEEEproof}

We note here that this transformation lowers the access ratio compared to the original code \eqref{eq34}, but in the mean time increases the average updates for each systematic element. According to different system requirements, one can choose one code over another.

The transformation in this section provides a general method to trade updates for access. Given any optimal-bandwidth code, one can define such transformations and manipulate the encoding matrices  to lower the access ratio.

%
\section{Concluding Remarks} \label{sec5}
%
%
\begin{figure}
  \centering
  \includegraphics[width=.7\textwidth]{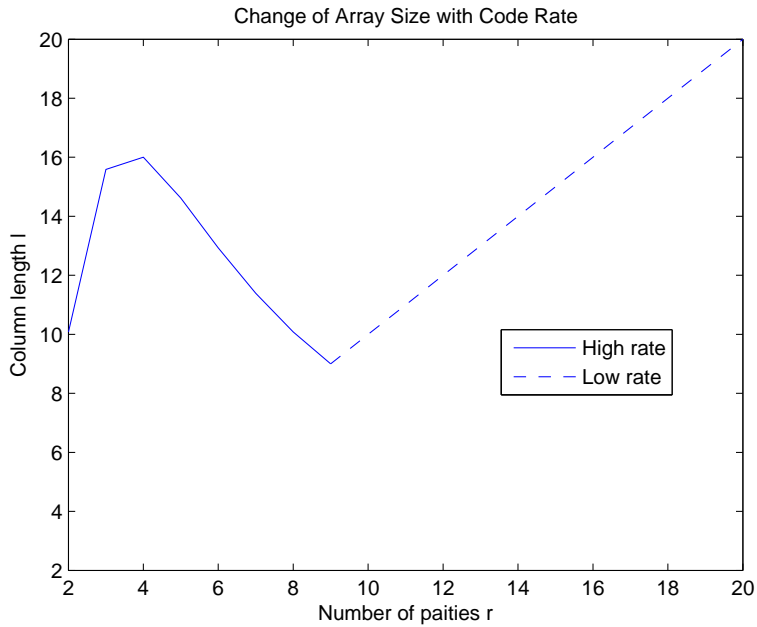}\\
  \caption{Change of array size with Code rate. $k=10$. For high code rate or $r \le 9$, the column length is shown in the solid line. For low code rate or $r \ge 9$, the column length is shown in the dashed line.}\label{fig5}
\end{figure}

In this paper, we presented a family of codes with parameters $(n=(r+1)m+r,k=(r+1)m,l=r^m)$ and they are the longest known high-rate MDS code with optimal repair. The codes were constructed using eigenspaces of the encoding matrices, such that they satisfy the subspace property. This property gives more insights on the structure of the codes, and simplifies the proof of optimal repair.

If we require that the code rate approaches $1$, i.e., $r$ being a constant and $m$ goes to infinity, then the column length $l$ is \emph{exponential} in the code length $k$. However, if we require the code rate to be roughly a constant fraction, i.e., $m$ being a constant and $r$ goes to infinity, then $l$ is \emph{polynomial} in $k$. Therefore, depending on the application and therefore the different codes rate, one can obtain different asymptotic characteristics of the code length.

For $n \ge 2k$ or $k \le r$ (low code rate), constructions in \cite{Suh-alignment,Rashmi} give the column length $l = r$. With some modifications, this column length is feasible for all $k \le r+1$. In our construction (high code rate), the column length is $l = r ^ {\frac{k}{r+1}}$. Fix the value of $k$, we can draw the graph of the column length with respect to the number of parities. Even though we need integer values for $k,r,l$, this graph still shows the trend of the code parameters. For example, this relationship is shown in Figure \ref{fig5} for $k=10$. These two regimes coincide when $r=k-1=9$. Actually, we can see that these two constructions are identical for $r=k-1$. Note that our construction only considers the repair of systematic nodes, so is only practical when $k >> r+1$. It is interesting to investigate the actual shape of this curve, and to understand for fixed code length $k$ how the column length $l$ changes with the number of parities $r$.

Besides, one possible application of the codes is hot/cold data. Since some of the nodes have lower access ratio than others if erased and hot data is more commonly requested, we can put the hot data in the low-access nodes, and cold data in the others.

At last, it is still an open problem what is the longest optimal-repair code one can build given the column length $l$. Also, the bound of the finite field size used for the codes may not be tight enough. Unlike the constructions in this paper, the field size may be reduced when we assume that the encoding matrices do not have eigenvalues or eigenvectors (are not diagonalizable).


\end{document}